\documentclass[aps,floats,twocolumn,prl,nofootinbib]{revtex4-2}
\usepackage{amsfonts,amsmath,amssymb,ascmac,bm,amsthm}
\usepackage{fnpct} 
\usepackage{comment}
\usepackage{ifpdf}
\usepackage{graphicx}
\usepackage{slashed}
\usepackage{color}
\usepackage[mathscr]{eucal}
\usepackage[utf8]{inputenc}
\usepackage{physics}
\usepackage{cancel}
\usepackage{float}
\usepackage{soul}
\usepackage{booktabs}
\usepackage{simpler-wick}
\usepackage{hyperref}
\usepackage{tensor}
\usepackage{upgreek}
\usepackage{caption}
\usepackage{subfigure}
\usepackage{subcaption}
\usepackage{tikz}
\usepackage{appendix}

\newcommand\mi{\mathrm{i}}
\newcommand\me{\mathrm{e}}

\newcommand{\dif}{\mathrm{d}}

\begin{document}
\newcommand{\Spec}{\operatorname{Spec}}
\newtheorem{theorem}{Theorem}

\title{Symmetry preservation in black hole quasinormal mode spectra}

\author{Han-Wen Hu$^{1,2}$}
\email{huhanwen@itp.ac.cn}
\author{Chen Lan$^{3}$}
\email{stlanchen@126.com}
\author{Zong-Kuan Guo$^{1,2,4}$}
\email{guozk@itp.ac.cn}
\author{Rong-Gen Cai$^{5}$}
\email{caironggen@nbu.edu.cn}
\affiliation{$^1$Institute of Theoretical Physics, Chinese Academy of Sciences, P.O. Box 2735, Beijing 100190, China}
\affiliation{$^2$School of Physical Sciences, University of Chinese Academy of Sciences, No.19A Yuquan Road, Beijing 100049, China}
\affiliation{$^3$Department of Physics, Yantai University, 30 Qingquan Road, Yantai 264005, China}
\affiliation{$^4$School of Fundamental Physics and Mathematical Sciences, Hangzhou Institute for Advanced Study, University of Chinese Academy of Sciences, Hangzhou 310024, China}
\affiliation{$^5$Institute of Fundamental Physics and Quantum Technology, \& School of Physical Science and Technology, Ningbo University, Ningbo 315211, China}

\begin{abstract}
We establish a general relation between symmetries of gravitational theories and black hole (BH) quasinormal mode (QNM) spectra. 
We show that a symmetry implies isospectrality when it induces a bijection between the corresponding QNM boundary value problems.
However, conformally related BHs have been reported to exhibit both conformal factor dependent and independent QNM spectra.
To resolve this issue, we develop a reduction scheme for higher order perturbation equations. 
Applied to pure Weyl gravity, it yields the complete axial spectrum of Schwarzschild, including Regge-Wheeler and spin-1 branches.
Our results show that these seemingly contradictory conclusions result from differences in the perturbation dynamics or in the boundary conditions.
\end{abstract}

\maketitle

\noindent
{\bf \emph{Introduction.--}}
Symmetries constrain the perturbation equations of BHs and can also establish spectral correspondences between different perturbation sectors or backgrounds.
QNMs of BHs are discrete resonances of linear perturbations subject to specified boundary conditions \cite{Kokkotas:1999bd,Berti:2009kk}, and their frequencies are poles of the retarded Green's function \cite{Leaver:1986gd,Hui:2019aox,Jaramillo:2020tuu}.
The isospectrality of axial and polar perturbations of a Schwarzschild BH can be understood through a Darboux transformation, while the Chandrasekhar duality can be promoted to an off-shell symmetry of the action for linearized Schwarzschild perturbations \cite{Glampedakis:2017rar,Solomon:2023ltn}.
Existing work has further examined whether corrections involving higher derivatives in gravitational theories break isospectrality between the axial and polar sectors \cite{Cano:2024wzo,Bah:2026aia}.
More recently, chiral structures have been used to derive sufficient conditions for isospectrality in more general perturbation systems \cite{Weller:2026jqu}.
These studies mainly concern spectral degeneracies between the axial and polar sectors, or between independent perturbation branches, on the same background.
A symmetry transformation that leaves the action invariant can also relate different BH solutions of the same theory and provide a map between their perturbation equations.
For BHs related by such a transformation, equality of their QNM spectra depends on whether the perturbation map contains all physical perturbations and preserves the QNM boundary conditions and the normalization of the time coordinate.
In gravity theories with higher derivatives, additional degrees of freedom change the solution space of the linearized equations \cite{Tattersall:2017erk,Cardoso:2019mqo,McManus:2019ulj,Cano:2023tmv,Blazquez-Salcedo:2016enn}, and generate new frequency branches and pole structures \cite{Zinhailo:2018ska,Antoniou:2024jku,Konoplya:2025afm,Cano:2020cao,Konoplya:2022iyn,Antoniou:2026nhh}; the perturbation map induced by a symmetry transformation must therefore include these additional branches.

Four dimensional pure Weyl gravity provides a concrete example of this problem.
Its action has local conformal symmetry, its vacuum field equation is the Bach equation, and it propagates degrees of freedom beyond the Einstein sector \cite{Mannheim:2011ds,Lu:2011ks,Deser:2012qg,Bergshoeff:2011ri}.
The Mannheim-Kazanas family gives static, spherically symmetric Bach vacuum geometries \cite{Mannheim:1988dj,Riegert:1984zz}, whose metrics are locally conformal to Schwarzschild-(A)dS Einstein metrics \cite{Schmidt:1999vp,Bambi:2016wdn}.
Thus, local conformal transformations relate a family of BH solutions within the same action, while the QNM spectrum on each background is determined jointly by the corresponding Bach perturbation equation and boundary conditions.
Existing QNM calculations for conformally related BHs have considered either test fields with different conformal properties \cite{Momennia:2018hsm,Momennia:2019cfd,Konoplya:2020fwg,Becar:2023jtd,Konoplya:2025mvj,Lutfuoglu:2025hjy,Toshmatov:2017bpx}, effective Einstein dynamics after conformal symmetry breaking \cite{Chen:2019iuo,Liu:2020ddo}, or a second-order Regge-Wheeler type equation on a particular metric \cite{Momennia:2019edt}.
These calculations concern different dynamical systems or spectral problems.
A comparison of QNM spectra on conformally related backgrounds therefore requires the complete axial spectrum determined by the linearized Bach equation and a check that the conformal transformation preserves the same boundary value problem.

To address this problem, we first prove a QNM spectral theorem applicable to general symmetry transformations and give sufficient conditions for isospectrality between BHs related by symmetry transformations and for preservation of the resonance pole orders.
We then apply the theorem to pure Weyl gravity by taking a Schwarzschild-(A)dS metric as the seed, introducing the linearized Einstein tensor as an intermediate variable, and reducing the Bach perturbation equation to two coupled second-order radial equations.
We construct the map between the master variables before and after the conformal transformation, derive the corresponding relation between the radial Bach operators, and check whether the boundary conditions and the normalization of the time coordinate are preserved \cite{Langlois:2021xzq}.
For a conformal transformation satisfying the theorem, the spectrum of the Schwarzschild representative determines the common spectrum of the corresponding BH family.

On a Schwarzschild background, the complete axial spectrum contains the Regge-Wheeler and spin-1 branches, and the latter shows that a single Regge-Wheeler equation cannot give the complete axial perturbation spectrum of pure Weyl gravity.
The repeated spin-2 degree of freedom makes the Regge-Wheeler frequencies second-order poles \cite{Lu:2011zk,Porrati:2011ku}.
The theorem can likewise be used to determine whether BH solutions related by other continuous or discrete symmetry transformations have the same QNM spectrum.
The reduction scheme can also be used for perturbations in higher derivative theories with similar structures \cite{Mauro:2015waa,Antoniou:2024jku}.

\vspace{5pt}
\noindent
{\bf \emph{Spectral Symmetry.--}}
We first prove this theorem and then use it to analyze axial perturbations in pure Weyl gravity.
\begin{theorem}\label{thm:qnm-spectral-symmetry}
For two BH solutions of the same gravitational theory related by a symmetry transformation, if the complete physical perturbation map induced by the transformation is a bijection between the domains of the two QNM boundary value problems, and the transformation preserves time translations and their normalization, then the two BHs have the same QNM spectrum.
If the maps induced by the transformation, together with their inverses, are analytic in the frequency domain under consideration, the pole locations and orders of the corresponding retarded Green's functions are also the same.
\end{theorem}

\begin{proof}
We use $\varphi^A$ to denote the metric and all dynamical fields, and write the field equations as $\mathscr{F}_A[\varphi]\equiv\delta S/\delta\varphi^A=0$.
The symmetry transformation $\tilde{\varphi}=s[\varphi]$ satisfies $S[\tilde{\varphi}]=S[\varphi]$; the derivation below is unchanged if the two sides differ by a boundary term that does not affect the field equations.
Decomposing the inverse transformation into parts without and with field derivatives, its linear variation can be written as
\begin{equation}\label{eq:derivative-transformation}
 \delta\varphi^B
 =K^B{}_A\delta\tilde{\varphi}^A
 +K^{B\mu}{}_A\partial_{\mu}\delta\tilde{\varphi}^A+\cdots,
\end{equation}
where all $K$'s are fixed by the inverse transformation.
Varying the action before and after the transformation in the same coordinates and integrating the derivative terms by parts gives
\begin{align}\label{eq:eom-symmetry-map}
 \mathscr{F}_A[\tilde{\varphi}]
 &=K^B{}_A\mathscr{F}_B
 -\partial_{\mu}\!\left(K^{B\mu}{}_A\mathscr{F}_B\right)+\cdots \nonumber\\
 &\equiv\left(\mathcal{W}_s[\varphi]\mathscr{F}[\varphi]\right)_A,
\end{align}
where the $\mathscr{F}_B$ on the right-hand side are evaluated on the original fields, and $\mathcal{W}_s$ denotes the induced action on the field equations.

Let the two BH backgrounds be $\bar{\varphi}$ and $\bar{\varphi}_s=s[\bar{\varphi}]$.
Equation~\eqref{eq:eom-symmetry-map} shows that $\mathscr{F}[\bar{\varphi}_s]=0$ whenever $\mathscr{F}[\bar{\varphi}]=0$; the inverse transformation gives the converse.
Add a perturbation $h$ to the original background and define
\begin{equation}\label{eq:linearized-symmetry-map}
 s[\bar{\varphi}+\epsilon h] \equiv \bar{\varphi}_s+\epsilon\mathcal{U}_s h+O(\epsilon^2).
\end{equation}
Let $\mathcal{L}_0$ and $\mathcal{L}_s$ denote the complete linearized operators on the two backgrounds, so that $\mathscr{F}[\bar{\varphi}+\epsilon h]=\epsilon\mathcal{L}_0h+O(\epsilon^2)$ and $\mathscr{F}[\bar{\varphi}_s+\epsilon\tilde{h}]=\epsilon\mathcal{L}_s\tilde{h}+O(\epsilon^2)$.
Substituting Eq.~\eqref{eq:linearized-symmetry-map} into Eq.~\eqref{eq:eom-symmetry-map} and comparing terms of first order in $\epsilon$ gives $\mathcal{L}_s\mathcal{U}_s=\mathcal{W}_s[\bar{\varphi}]\mathcal{L}_0$, because $\mathscr{F}[\bar{\varphi}]=0$ makes the variation of $\mathcal{W}_s$ vanish.

Since the symmetry transformation preserves time translations and their normalization, a perturbation with frequency $\omega$ still evolves as $\me^{-\mi\omega t}$ after the transformation.
Writing $\mathcal{U}_s(\me^{-\mi\omega t}\Phi)=\me^{-\mi\omega t}\mathcal{T}_s(\omega)\Phi$, the preceding relation becomes in the frequency domain
\begin{equation}\label{eq:frequency-symmetry-map}
 \mathcal{L}_s(\omega)\mathcal{T}_s(\omega)
 =\mathcal{W}_s(\omega)\mathcal{L}_0(\omega).
\end{equation}
The perturbation spaces retain all linearized constraints, and solutions differing only by allowed pure gauge perturbations are identified as the same physical solution.
If $\Phi$ is a nonzero physical solution satisfying the QNM boundary conditions and $\mathcal{L}_0(\omega)\Phi=0$, Eq.~\eqref{eq:frequency-symmetry-map} shows that $\mathcal{T}_s(\omega)\Phi$ satisfies the transformed perturbation equation and the corresponding boundary conditions.
The inverse transformation gives the converse correspondence, so the physical QNMs of the two BHs are in bijection and
\begin{equation}\label{eq:qnm-spectrum}
 \Spec_{\rm QNM}(\mathcal{L}_s)
 =\Spec_{\rm QNM}(\mathcal{L}_0).
\end{equation}
Here $\Spec_{\rm QNM}(\mathcal{L})$ denotes the QNM spectrum of the perturbation operator $\mathcal{L}$, namely the set of frequencies $\omega$ for which $\mathcal{L}(\omega)\Phi=0$ has a nonzero physical solution satisfying the QNM boundary conditions.
Because the perturbation map is invertible, the number of independent physical modes at each frequency is also unchanged.

For the sourced equation $\mathcal{L}_0\Phi=J$, Eq.~\eqref{eq:frequency-symmetry-map} also gives the source transformation $\tilde{J}=\mathcal{W}_sJ$.
Repeating the derivation for the inverse transformation shows that $\mathcal{W}_s$ is likewise invertible on the corresponding equation and source spaces.
Let $\mathcal{G}_0=\mathcal{L}_0^{-1}$ and $\mathcal{G}_s=\mathcal{L}_s^{-1}$ be the retarded Green's functions satisfying the corresponding boundary conditions. Then
\begin{equation}\label{eq:general-resolvent-map}
 \mathcal{G}_s(\omega)
 =\mathcal{T}_s(\omega)\mathcal{G}_0(\omega)\mathcal{W}_s^{-1}(\omega).
\end{equation}
If $\mathcal{T}_s$, $\mathcal{W}_s$, and their inverses are analytic in the frequency domain under consideration, they generate no new poles and cannot remove the highest-order coefficient of an existing pole. 
The two retarded Green's functions therefore have the same pole locations and orders.
\end{proof}

We now apply the theorem to Weyl gravity.
We first construct the complete axial Bach perturbations and their conformal map on a general Schwarzschild-(A)dS-type Einstein metric; we then use the Ricci-flat Schwarzschild metric as a representative to determine the QNM spectrum of the conformally related BHs that satisfy the theorem's conditions.

\vspace{5pt}
\noindent
{\bf \emph{Axial Perturbation and Conformal Mapping.--}}
In four-dimensional spacetime, the action of Weyl gravity and its vacuum field equation are
\begin{equation}\label{eq:weyl-action}
 S_{\rm W}=\int \dif^4x\sqrt{-g}\;  C_{\mu\nu\rho\sigma}C^{\mu\nu\rho\sigma},
 \quad \mathcal{B}_{\mu\nu}=0,
\end{equation}
where $\mathcal{B}_{\mu\nu}$ is the Bach tensor.
Static, spherically symmetric Bach vacuum solutions are described by the Mannheim-Kazanas family and can locally be written as conformal transformations of Schwarzschild-(A)dS Einstein metrics.
The local solution space of the complete Bach perturbation problem can therefore be obtained by a conformal transformation, so we first study the Schwarzschild-(A)dS metric
\begin{align}\label{eq:einstein-representative}
 	\dif s_{\rm Sch}^2 & = -F(\rho)\dif t^2+\frac{\dif\rho^2}{F(\rho)}+\rho^2\dif\Omega_2^2, \nonumber \\
 	F(\rho) & = 1-\frac{2M}{\rho}+\frac{\hat{\Lambda}}{3}\rho^2,
\end{align}
and consider its conformal transformation $\tilde{g}_{\mu\nu}=C^2(\rho)g^{\rm Sch}_{\mu\nu}$, where $C(\rho)$ is smooth, finite, and nonzero throughout the BH exterior.

In the Regge-Wheeler gauge, axial perturbations are parametrized by $h_{tA}=\me^{-\mi\omega t}h_0(\rho)S_A^{\ell m}$ and $h_{\rho A}=\me^{-\mi\omega t}h_1(\rho)S_A^{\ell m}$, where $A,B$ label angular coordinates, $\omega$ is the complex frequency, $Y_{\ell m}$ is a scalar spherical harmonic, and $S_A^{\ell m}=\epsilon_A{}^B D_B Y_{\ell m}$.
Here $D_A$ and $\epsilon_{AB}$ are the covariant derivative and volume form on the unit two-sphere, respectively.
Introducing the tortoise coordinate $\dif r_*/\dif\rho=F^{-1}$, we define the master variables in the Regge-Wheeler gauge,
\begin{equation}\label{eq:master-variables}
 \Psi=\frac{F}{\rho}h_1, \quad \Xi=\frac{\mi\omega h_0+\partial_{r_*}(\rho\Psi)}{F}.
\end{equation}
For $\omega\neq0$, $h_0$ and $h_1$ are uniquely determined by Eq.~\eqref{eq:master-variables}; hence $\Psi$ and $\Xi$ contain the full axial metric perturbation.

To separate the fourth-order Bach perturbation equation, we start from the linearized Einstein equation and define
\begin{equation}\label{eq:linearized-einstein-tensor}
 \mathcal{E}_{\mu\nu}[h] \equiv\delta G_{\mu\nu}[h]-\hat{\Lambda}h_{\mu\nu},
\end{equation}
where $h_{\mu\nu}$ is the metric perturbation and $\delta G_{\mu\nu}$ is the linearized Einstein tensor.
We decompose the Bach tensor using the Schouten tensor $P_{\mu\nu}=(R_{\mu\nu}-Rg_{\mu\nu}/6)/2$.
Together with the tracelessness of axial perturbations and the linearized Bianchi identity, this allows the linearized Schouten tensor to be written entirely in terms of $\mathcal{E}_{\mu\nu}$.
The derivation in Appendix~\ref{app:linearized-bach} gives
\begin{align}\label{eq:linearized-bach-image}
	 \delta\mathcal{B}_{\mu\nu} = \frac{1}{2}\left[ \left(\nabla^2+\frac{4\hat{\Lambda}}{3}\right)\mathcal{E}_{\mu\nu}+2C_{\mu\rho\nu\sigma}\mathcal{E}^{\rho\sigma}\right].
\end{align}
Here $\nabla_\mu$ and $C_{\mu\nu\rho\sigma}$ are constructed from the background metric $g_{\mu\nu}^{\rm Sch}$.
The fourth derivatives in the Bach perturbation therefore enter through the second-order tensor $\mathcal{E}_{\mu\nu}[h]$.

Separating Eq.~\eqref{eq:linearized-bach-image} gives two radial operators,
\begin{subequations}
	\begin{align}\label{eq:radial-operators}
		\mathcal{R}_1 & =\partial_{r_*}^2+\omega^2
		-F\frac{\ell(\ell+1)}{\rho^2}, \\
		\mathcal{R}_2 & =\partial_{r_*}^2+\omega^2
		-F\left[\frac{\ell(\ell+1)}{\rho^2}-\frac{6M}{\rho^3}\right].
	\end{align}
\end{subequations}
Here $\mathcal{R}_2$ is the axial Regge-Wheeler operator, while $\mathcal{R}_1$ has the same form as the radial operator for an axial Maxwell field.
We further define
\begin{equation}\label{eq:q-definition}
 Q=F^{-1}\mathcal{R}_2\Psi-\frac{2M}{\rho^3}\Xi.
\end{equation}
Substituting the spherical-harmonic components of $\mathcal{E}_{\mu\nu}$ into Eq.~\eqref{eq:linearized-bach-image} and eliminating dependent components with the linearized Bianchi identity puts the two independent Bach perturbation equations in the form
\begin{equation}\label{eq:two-channel}
 \begin{pmatrix}
  \mathcal{R}_1 & -2F \\[2pt]
  -\dfrac{8M\hat{\Lambda}F}{3\rho^3} &
  \mathcal{R}_2-\dfrac{2\hat{\Lambda}F}{3}
 \end{pmatrix}
 \begin{pmatrix}
  \Xi \\ Q
 \end{pmatrix}
 =0.
\end{equation}
Equations~\eqref{eq:q-definition} and \eqref{eq:two-channel} form the complete master equation for Bach perturbations.
For $Q=\Xi=0$, the system reduces to $\mathcal{R}_2\Psi=0$, so the usual Regge-Wheeler perturbations are embedded in the Bach solution space.
A nonzero reduced cosmological constant $\hat{\Lambda}$ couples the two radial branches through the $M\hat{\Lambda}/\rho^3$ term in the lower-left entry of the matrix.

Let $\Phi=(\Psi,\Xi)^{\rm T}$ and denote the radial Bach operator acting on $\Phi$ by $\mathcal{L}_{\rm B}(\omega)$.
We next construct the map between the master variables before and after the conformal transformation.
Using the same coordinates $t$ and tortoise coordinate $r_*$ in the two metrics, the axial perturbation components satisfy $\tilde{h}_0=C^2 h_0$ and $\tilde{h}_1=C^2 h_1$.
Writing $q=C^{\prime}/C$, where $^{\prime}\equiv\dif/\dif\rho$, Eq.~\eqref{eq:master-variables} gives
\begin{equation}\label{eq:conformal-map}
 \tilde{\Phi}=\mathcal{T}_C\Phi,
 \quad
 \mathcal{T}_C=C^2
 \begin{pmatrix}
  1 & 0 \\
  2\rho q & 1
 \end{pmatrix}.
\end{equation}
Since $\det\mathcal{T}_C=C^4$, this lower-triangular map is invertible wherever $C\neq0$.

After arranging the transformed radial Bach equation with the normalization of the seed equation, denote its operator by $\tilde{\mathcal{L}}_{\rm B}$; in Eq.~\eqref{eq:frequency-symmetry-map}, set $\mathcal{T}_s=\mathcal{T}_C$ and $\mathcal{W}_s=1$, so that
\begin{equation}\label{eq:intertwining}
 \tilde{\mathcal{L}}_{\rm B}(\omega)\mathcal{T}_C=\mathcal{L}_{\rm B}(\omega).
\end{equation}
Thus, for general $\hat{\Lambda}$, Eq.~\eqref{eq:conformal-map} maps the complete axial perturbations before and after the conformal transformation, while Eq.~\eqref{eq:intertwining} maps the corresponding Bach equations.
By Theorem~\ref{thm:qnm-spectral-symmetry}, whether the QNM spectra before and after the conformal transformation agree further depends on whether $\mathcal{T}_C$ preserves the corresponding boundary conditions and the normalization of the time coordinate.

\vspace{5pt}
\noindent
{\bf \emph{Schwarzschild Conformal Spectrum.--}}
We now set $\hat{\Lambda}=0$, so the seed metric is the asymptotically flat Schwarzschild metric with $F=1-2M/\rho$.
The QNM boundary conditions are purely ingoing at the event horizon and purely outgoing at spatial infinity.
When $C$ and $C^{-1}$ are smooth and finite in the exterior and at both boundaries, and $\rho q\to0$ at infinity, $\mathcal{T}_C$ is finite and invertible at the horizon and infinity: it changes only the amplitudes of the ingoing and outgoing solutions and does not mix the two boundary behaviors.
If the two metrics use the same normalization of the time coordinate, the conditions of Theorem~\ref{thm:qnm-spectral-symmetry} are satisfied.
It is therefore sufficient to determine the spectrum of the Schwarzschild representative; this gives the common spectrum of all BHs related to it by such conformal transformations.

For the asymptotically flat Schwarzschild metric, Eq.~\eqref{eq:two-channel} becomes
\begin{equation}\label{eq:triangular-system}
 \mathcal{R}_1 \Xi = 2 F Q, \quad \mathcal{R}_2 Q = 0.
\end{equation}
The second equation gives a spin-2 degree of freedom, while the first is a spin-1 equation driven by $Q$.
To separate the driving term, define the differential operators
\begin{align}\label{eq:off-shell-identity}
 \mathcal{D} & = \frac{1}{\omega^2}\left[\frac{2M}{\rho}-\frac{\ell(\ell+1)+1}{3}+(2M-\rho)\partial_\rho
 \right], \nonumber \\
 \mathcal{K} & = \frac{1}{3\rho\omega^2}\Bigl\{3\rho(\rho-2M)\partial_\rho + [\ell(\ell+1)+7]\rho-18M\Bigr\}.
\end{align}
They obey the operator identity $\mathcal{R}_1\mathcal{D}-2F =-\mathcal{K}\mathcal{R}_2$.
With $Z_1=\Xi-\mathcal{D}Q$, Eq.~\eqref{eq:triangular-system} becomes
\begin{equation}\label{eq:decoupled-system}
 \mathcal{R}_2 Q=0, \quad \mathcal{R}_1 Z_1=0.
\end{equation}
The operator $\mathcal{D}$ contains only finitely many derivatives and therefore preserves the QNM boundary conditions for $\omega\neq0$.
On a Ricci-flat background, the six degrees of freedom of pure Weyl gravity in four dimensions are represented by two massless spin-2 degrees of freedom and one massless spin-1 degree of freedom \cite{Deser:2012qg}.
The kernel of $\mathcal{E}$ gives the Regge-Wheeler tensor modes, $Q$ describes the second tensor modes in the non-Einstein sector, and $Z_1$ isolates the vector degree of freedom contained in $\Xi$.
The tidal curvature $M/\rho^3$ of the Schwarzschild spacetime provides Regge-Wheeler and Maxwell-type potential barriers for the tensor and vector degrees of freedom, respectively; for $\hat{\Lambda}=0$, it does not couple the two branches.

Equation~\eqref{eq:decoupled-system} gives the inclusions between the spectra.
For any Bach QNM, both $Q$ and $Z_1$ satisfy the QNM boundary conditions.
If at least one of them is nonzero, its frequency belongs to the spectrum of the corresponding second-order operator; if $Q=Z_1=0$, then $\Xi=0$ and Eq.~\eqref{eq:q-definition} reduces to $\mathcal{R}_2\Psi=0$.
Therefore $\Spec_{\rm QNM}^{\rm Bach}\subseteq \Spec_{\rm QNM}^{\rm RW}\cup\Spec_{\rm QNM}^{s=1}$.
For a physical branch satisfying the Regge-Wheeler equation $\mathcal{R}_2Z_2=0$, we can directly choose $(\Psi,\Xi)_2=(Z_2,0)$.
Likewise, after setting $Q=0$, a solution of $\mathcal{R}_1Z_1=0$ reconstructs the metric perturbation as
\begin{equation}\label{eq:two-lifts}
 (\Psi,\Xi)_1=\left(\frac{1}{3}Z_1,Z_1\right).
\end{equation}
Because both reconstruction maps preserve the QNM boundary conditions, the reverse inclusion holds.
For $\ell\geq2$ and $\omega\neq0$, excluding algebraically special frequencies, the complete QNM spectrum is therefore the union of the two branches,
\begin{equation}\label{eq:spectral-union}
 \Spec_{\rm QNM}^{\rm B} =\Spec_{\rm QNM}^{\rm RW}\cup\Spec_{\rm QNM}^{s=1}.
\end{equation}
Since $\mathcal{T}_C$ preserves these boundary conditions, Eq.~\eqref{eq:spectral-union} applies to all BH solutions related by conformal transformations satisfying the conditions above, namely
\begin{equation}\label{eq:conformal-spectrum}
 \Spec_{\rm QNM}(\tilde{\mathcal{L}}_{\rm B})
 =\Spec_{\rm QNM}(\mathcal{L}_{\rm B})
 =\Spec_{\rm QNM}^{\rm RW}\cup\Spec_{\rm QNM}^{s=1}.
\end{equation}
This shows that the spin-1 branch is excited entirely by metric perturbations and that the QNM spectrum of pure Weyl gravity on a Schwarzschild background contains two independent frequency branches.
A single Regge-Wheeler type equation therefore determines only the tensor branch of the Bach spectrum, not the complete axial spectrum of pure Weyl gravity.

The second tensor degree of freedom added by the theory with higher derivatives does not produce a new set of QNM frequencies; it makes the existing Regge-Wheeler resonances second-order poles.
Let the resolvent of the Bach operator acting on the master variables be $\mathcal{G}_{\rm B}=\mathcal{L}_{\rm B}^{-1}$.
Near a Regge-Wheeler pole $\omega_n$, the metric resolvent has the expansion
\begin{equation}\label{eq:double-pole}
 \mathcal{G}_{\rm B}(\omega) \sim \frac{G_{-2}^{(n)}}{(\omega-\omega_n)^2} + \frac{G_{-1}^{(n)}}{\omega-\omega_n}, \quad\ G_{-2}^{(n)} \neq 0,
\end{equation}
where $G_{-p}^{(n)}$ is the coefficient of the $p$th-order pole of the resolvent at $\omega_n$.
The spin-2 degrees of freedom therefore appear as generalized QNMs, in agreement with the critical structure of gravity defined by the square of the Weyl tensor \cite{Lu:2011ks}.
The contribution of these poles to the time domain perturbation is $h_n(t) \sim (A_n + B_n t)\me^{-\mi \omega_n t}$.
For $\Im\omega_n<0$, the perturbation still decays exponentially, but the linear time factor distinguishes it from an ordinary exponentially damped mode.
The corresponding resolvent relation follows directly from Eq.~\eqref{eq:general-resolvent-map},
\begin{equation}\label{eq:resolvent-map}
 \tilde{\mathcal{G}}_{\rm B}(\omega)=\mathcal{T}_C\mathcal{G}_{\rm B}(\omega).
\end{equation}
Because $\mathcal{T}_C$ and its inverse are nonsingular in the BH exterior and at the boundaries and are independent of $\omega$, the locations and orders of all poles in the Bach spectrum are unchanged; in particular, the second-order poles of the Regge-Wheeler branch are not altered by the conformal transformation.
Appendix~\ref{app:numerical-verification} also solves the Bach perturbation equations and the conformally transformed radial equations directly, providing an independent numerical check of Eqs.~\eqref{eq:spectral-union} and \eqref{eq:conformal-spectrum}.
If the conformal factor diverges or tends to zero at a boundary, or changes the asymptotic boundary structure, the boundary conditions no longer satisfy the requirements of the theorem; the conformal transformation between Schwarzschild-(A)dS and Mannheim-Kazanas metrics is an example.

\vspace{5pt}
\noindent
{\bf \emph{Conclusion and Discussion.--}}
In this letter, we established a general relation between symmetries and QNM spectra of BHs, showing how isospectrality follows from a map between complete QNM boundary value problems.
In pure Weyl gravity, we determined the complete axial spectrum on a Schwarzschild background from the linearized Bach equation.
The spectrum contains the Regge-Wheeler and spin-1 frequency branches, and the repeated spin-2 factor makes the Regge-Wheeler frequencies second order poles.
This spectrum and its pole structure are shared by conformally related BHs that satisfy the theorem.
Einstein gravity and pure Weyl gravity share the same Regge-Wheeler branch on a Schwarzschild background because every solution of the linearized Einstein equation also gives a Bach perturbation.
A common background metric can preserve part of the dynamics, but it does not by itself determine the complete QNM spectrum.

Conformal transformations in pure Weyl gravity are only one example of the theorem.
In Einstein-Maxwell and Einstein-Abelian-Higgs theories, local $U(1)$ transformations relate different gauge representatives of the same charged BH; after pure gauge modes are removed, they give the same physical perturbation space \cite{Gubser:2008px}.
The Einstein-Maxwell action is also invariant under the discrete transformation $A_\mu\to-A_\mu$, which maps a Kerr-Newman solution with charge $Q$ to one with charge $-Q$ and changes the sign of the electromagnetic perturbation \cite{Newman:1965my}.
For a Kerr BH with a complex scalar hair, a global $U(1)$ phase rotation acts simultaneously on the background scalar and its perturbation, so the perturbation map is only a constant phase factor \cite{Herdeiro:2014goa}.
In a linear scalar-Gauss-Bonnet coupling, the shift $\varphi\to\varphi+c$ changes the action only by a Gauss-Bonnet topological term, leaving the field equations and their linearization unchanged; it relates hairy BHs whose asymptotic scalar values differ by $c$ \cite{Sotiriou:2013qea,Sotiriou:2014pfa}.
The action with a quadratic coupling is invariant under the $\mathbb{Z}_2$ transformation $\varphi\to-\varphi$, which relates two scalarized BHs with opposite scalar charges and induces $\delta\varphi\to-\delta\varphi$ at the linearized level \cite{Doneva:2017bvd,Silva:2017uqg}.
These continuous or discrete symmetry transformations relate BH solutions with gauge fields or scalar hair and establish bijections between the complete physical perturbation spaces in the corresponding QNM boundary value domains, while preserving the normalization of the time coordinate.
Each pair of mutually mapped BHs therefore has the same QNM spectrum.

The situation is different for the conformal transformation between a Mannheim-Kazanas solution and a Schwarzschild solution, it establishes only a local map between the Bach equations.
For $\hat{\Lambda}=0$, the conformal factor $C=(1-a\rho)^{-1}$ relating the two metrics diverges at $\rho_{\rm b}=1/a$ and maps spatial infinity of the Mannheim-Kazanas solution to a finite point in the Schwarzschild exterior.
This change of endpoint gives the two QNM boundary value problems different domains, so the local map between the Bach equations is not sufficient to establish isospectrality.
The role of boundary conditions in selecting dynamical sectors of conformal gravity has also been discussed in the literature \cite{Anastasiou:2016jix}.
At the quantum level, a Weyl anomaly may change the linearized field equations, so the map required by the spectral theorem must be reconsidered \cite{Duff:1993wm,Mazur:2001aa}.

The reduction scheme for the Bach equation can also be used for polar perturbations and extended to other theories with higher derivatives.
Related reduction procedures have been carried out by introducing auxiliary fields or by solving directly for the linearized Ricci tensor \cite{Mauro:2015waa,Myung:2018ete}.
In general, one selects tensor variables of lower differential order and with definite transformation properties from the linearized field equations or curvature tensors, and then performs a spherical-harmonic decomposition using the symmetries of the background.
After the Bianchi identities remove dependent components, if the curvature couplings and covariant derivatives generate no new independent tensor components, the original higher derivative equations can be organized into a finite set of radial equations of lower order.
Solving these equations and reconstructing the metric perturbations from the definitions of the intermediate variables retains the perturbations already present in the lower-order theory, the additional degrees of freedom introduced by higher derivatives, and their QNM boundary conditions.
For example, in nondegenerate $f(R)$ gravity on a Schwarzschild background, we can choose $X=\delta R$ as the intermediate variable; 
it satisfies $(\bar\nabla^2-m_0^2)X=0$, which gives the massive spin-0 branch. When reconstructing the metric perturbation, the massless spin-2 homogeneous solution must also be retained, so the complete spectrum is the union of the two branches \cite{Myung:2011ih}.
For the more general $\mathcal{L}_{\rm quad}=R+\alpha R^2+\beta R_{\mu\nu}R^{\mu\nu}$, the traceless Ricci variable $Y_{\mu\nu}$ must also be introduced; $(\bar\Delta_{\rm L}+m_2^2)Y_{\mu\nu}=0$ gives an additional massive spin-2 branch, so the complete spectrum contains the two branches above together with this branch \cite{Antoniou:2024jku}.

\vspace{5pt}
\noindent
{\bf \emph{Acknowledgments.--}}
This work is supported by the National Natural Science Foundation of China No.~12475067 and No.~12235019. Moreover, C.Lan is supported by Yantai University under Grant No.~WL22B224.

\appendix

\section{Derivation of the linearized Bach equation}\label{app:linearized-bach}

This appendix derives Eq.~\eqref{eq:linearized-bach-image}.
An overbar denotes a quantity evaluated on the Einstein background.
With the conventions of Eq.~\eqref{eq:einstein-representative},
\begin{equation}\label{app:eq:einstein-background}
 \bar{R}_{\mu\nu}=-\hat\Lambda\bar{g}_{\mu\nu},\quad \bar{P}_{\mu\nu}=-\frac{\hat\Lambda}{6}\bar{g}_{\mu\nu}, \quad \bar{\nabla}_\alpha\bar{P}_{\mu\nu}=0.
\end{equation}
The last relation removes derivatives of the background Schouten tensor and prevents the associated variations of the connection from generating new independent tensor structures.

Axial perturbations are traceless and do not excite the even parity curvature scalar:
\begin{equation}\label{app:eq:axial-trace}
 h^\mu{}_{\mu}=0, \quad \delta R=0.
\end{equation}
Using $P_{\mu\nu}=(R_{\mu\nu}-Rg_{\mu\nu}/6)/2$ and $\mathcal{E}_{\mu\nu}=\delta G_{\mu\nu}-\hat\Lambda h_{\mu\nu}$, the Schouten perturbation becomes
\begin{equation}\label{app:eq:delta-p-e}
 \delta P_{\mu\nu}+\frac{\hat\Lambda}{6}h_{\mu\nu}=\frac{1}{2}\mathcal{E}_{\mu\nu}.
\end{equation}
This relation collects the fourth derivatives of the metric into the second order tensor $\mathcal{E}_{\mu\nu}$.

We first treat the derivative terms.
The identity for the variation of the connection,
$\bar{g}_{\nu\lambda}\delta\Gamma^\lambda_{\alpha\mu}
 +\bar{g}_{\mu\lambda}\delta\Gamma^\lambda_{\alpha\nu}
 =\bar{\nabla}_\alpha h_{\mu\nu}$ and Eq.~\eqref{app:eq:einstein-background} give
\begin{equation}\label{app:eq:nabla-p}
 \delta(\nabla_\alpha P_{\mu\nu})=\frac{1}{2}\bar{\nabla}_\alpha\mathcal{E}_{\mu\nu}.
\end{equation}
Write the Bach tensor in terms of the Schouten tensor as
\begin{equation}\label{app:eq:bach-schouten}
 \mathcal B_{\mu\nu}
 =\nabla^2P_{\mu\nu}-\nabla^\rho\nabla_\mu P_{\rho\nu}
 +P^{\rho\sigma}C_{\mu\rho\nu\sigma}.
\end{equation}
Since $\bar{\nabla}P=0$, variations of the outer inverse metric and connection multiply vanishing background derivatives.
Equation~\eqref{app:eq:nabla-p} then yields
\begin{align}
 \delta(\nabla^2P_{\mu\nu})
 &=\frac{1}{2}\bar{\nabla}^2\mathcal{E}_{\mu\nu},\nonumber\\
 \delta(\nabla^\rho\nabla_\mu P_{\rho\nu})
 &=\frac{1}{2}\bar{\nabla}^\rho
   \bar{\nabla}_\mu\mathcal{E}_{\rho\nu}.
 \label{app:eq:derivative-variations}
\end{align}

The last term in Eq.~\eqref{app:eq:bach-schouten} requires variations of both the raised Schouten tensor and the Weyl tensor.
Varying $g^{\rho\sigma}C_{\mu\rho\nu\sigma}=0$ and using Eq.~\eqref{app:eq:delta-p-e} gives
\begin{align}
 \bar{g}^{\rho\sigma}\delta C_{\mu\rho\nu\sigma}
 &=h^{\rho\sigma}\bar{C}_{\mu\rho\nu\sigma},\nonumber\\
 \delta P^{\rho\sigma}
 &=\bar{g}^{\rho\alpha}\bar{g}^{\sigma\beta}
   \delta P_{\alpha\beta}
   +\frac{\hat\Lambda}{3}h^{\rho\sigma}.
 \label{app:eq:raised-variations}
\end{align}
The metric perturbation in the first line combines exactly with the part in which the background Schouten tensor multiplies $\delta C_{\mu\rho\nu\sigma}$.
The curvature coupling reduces to
\begin{equation}\label{app:eq:weyl-coupling-variation}
 \delta(P^{\rho\sigma}C_{\mu\rho\nu\sigma})
 =\frac{1}{2}\bar{C}_{\mu\rho\nu\sigma}
  \mathcal{E}^{\rho\sigma}.
\end{equation}
Substituting Eqs.~\eqref{app:eq:derivative-variations} and \eqref{app:eq:weyl-coupling-variation} into Eq.~\eqref{app:eq:bach-schouten} gives
\begin{equation}\label{app:eq:bach-precommuted}
 \delta\mathcal B_{\mu\nu}
 =\frac{1}{2}\left(
 \bar{\nabla}^2\mathcal{E}_{\mu\nu}
 -\bar{\nabla}^\rho\bar{\nabla}_\mu\mathcal{E}_{\rho\nu}
 +\bar{C}_{\mu\rho\nu\sigma}\mathcal{E}^{\rho\sigma}
 \right).
\end{equation}

The linearized Bianchi identity and tracelessness remove the mixed derivative:
\begin{equation}\label{app:eq:e-constraints}
 \bar{\nabla}^\rho\mathcal{E}_{\rho\nu}=0,
 \qquad \mathcal{E}^\rho{}_{\rho}=0.
\end{equation}
On an Einstein background in four dimensions, the Riemann tensor splits into its Weyl and constant curvature parts, so
\begin{equation}\label{app:eq:commutator}
 \bar{\nabla}^\rho\bar{\nabla}_\mu\mathcal{E}_{\rho\nu}
 =-\bar{C}_{\mu\rho\nu\sigma}\mathcal{E}^{\rho\sigma}
  -\frac{4\hat\Lambda}{3}\mathcal{E}_{\mu\nu}.
\end{equation}
Substituting Eq.~\eqref{app:eq:commutator} into Eq.~\eqref{app:eq:bach-precommuted} gives
\begin{equation}
 \delta\mathcal B_{\mu\nu}
 =\frac{1}{2}\left[
 \left(\bar{\nabla}^2+\frac{4\hat\Lambda}{3}\right)\mathcal{E}_{\mu\nu}
 +2\bar{C}_{\mu\rho\nu\sigma}\mathcal{E}^{\rho\sigma}
 \right],
\end{equation}
which is Eq.~\eqref{eq:linearized-bach-image}.
The derivation uses only an Einstein background in four dimensions, the tracelessness of axial perturbations, and the linearized Bianchi identity.
For a background that does not satisfy the Einstein equations, $\bar{\nabla}_\alpha\bar{P}_{\mu\nu}$ is generally nonzero, Eq.~\eqref{app:eq:nabla-p} no longer holds, and additional terms from the variation of the connection enter the linearized Bach tensor.

\section{Numerical verification of the quasinormal frequencies}\label{app:numerical-verification}

This appendix tests the QNM spectrum derived in the main text.
We first solve the fourth order Bach equations directly in the variables $(\Psi,\Xi)$ with Leaver's method \cite{Leaver:1985ax}.
We then apply a regular conformal transformation and use a shooting method to verify conformal invariance.
Both calculations set $M=1$.
Reference frequencies are obtained independently by applying Leaver's method to the RW equation and the vector perturbation equation for a Schwarzschild BH.

For the Leaver calculation, let $x=1-2M/\rho$, so the event horizon and spatial infinity lie at $x=0$ and $x=1$, respectively.
Define
\begin{equation}\label{app:eq:leaver-prefactor}
 P_\omega=\exp\left(\frac{2\mi M\omega}{1-x}\right)x^{-2\mi M\omega}(1-x)^{-2\mi M\omega}.
\end{equation}
This factor gives $\me^{-\mi\omega r_*}$ at the horizon and $\me^{\mi\omega r_*}$ at infinity.
We expand the remaining parts about $x=0$:
\begin{equation}\label{app:eq:direct-leaver-series}
 \Psi=P_\omega \rho \sum_{n=0}^{\infty}\psi_nx^n,
 \qquad
 \Xi=P_\omega \rho \sum_{n=0}^{\infty}\xi_nx^n.
\end{equation}
Substituting Eq.~\eqref{app:eq:direct-leaver-series} directly into $B_1[\Psi,\Xi]=B_2[\Psi,\Xi]=0$ and collecting powers of $x$ gives a recurrence with seven terms,
\begin{equation}\label{app:eq:seven-band}
 \sum_{j=-4}^{2}\mathbf M_{n,j}(\omega)
 \begin{pmatrix}\psi_{n+j}\\ \xi_{n+j}\end{pmatrix}=0.
\end{equation}
At the horizon, the ingoing condition removes the divergent branches, and forward recurrence from the three remaining initial data produces a three dimensional ingoing solution space.
At infinity, the outgoing condition is equivalent to convergence of the series.
We select the three decaying branches from the large-$n$ asymptotic eigenvalues and recurse backward from a sufficiently high truncation order.
Matching the two three dimensional bases at an intermediate order gives a $6\times6$ Casoratian matrix $\mathbf C(\omega)$.
The Bach frequencies solve
\begin{equation}\label{app:eq:leaver-root}
 \Re\left[\det\mathbf{C}(\omega)\right]=0,
 \quad \Im\left[\det\mathbf{C}(\omega)\right]=0,
\end{equation}
which gives the QNM frequencies of the Bach boundary value problem.

For a direct numerical test of conformal invariance, we choose
\begin{equation}\label{app:eq:numerical-conformal-factor}
 C^2(\rho)=\left(1+\frac{L^2}{\rho^2}\right)^2, \quad L=M=1.
\end{equation}
For the transformed version of Eq.~\eqref{eq:decoupled-system}, define $Y_1=C^2Z_1$ and $Y_2=C^2Q$.
They satisfy
\begin{equation}\label{app:eq:shooting-ode}
 Y_s^{\prime\prime}+\left(\frac{F^{\prime}}{F}-4q\right)Y_s^{\prime} + A_s(\rho,\omega)Y_s=0,\quad s=1,\ 2,
\end{equation}
where
\begin{align}
 A_s={}&\frac{\omega^2}{F^2}-\frac{W_s}{F}
 +4q^2-2q^{\prime}-2\frac{F^{\prime}}{F}q,\nonumber\\
 W_1={}&\frac{\ell(\ell+1)}{\rho^2},\qquad
 W_2=\frac{\ell(\ell+1)}{\rho^2}-\frac{6M}{\rho^3}.
 \label{app:eq:shooting-coefficients}
\end{align}
We treat $\omega$ as an eigenvalue and solve these equations by numerical integration.

Let $z=\rho-2M$.
The ingoing expansion at the horizon and the outgoing expansion at infinity are
\begin{align}\label{app:eq:shooting-series}
 Y_s^{\rm H} & = z^{-2\mi M\omega}\sum_{k=0}^{N_{\rm H}}a_k z^k,\nonumber\\
 Y_s^\infty & = \me^{\mi\omega\rho}\rho^{2\mi M\omega} \sum_{k=0}^{N_\infty}\frac{b_k}{\rho^k}.
\end{align}
After setting $a_0=b_0=1$, Eq.~\eqref{app:eq:shooting-ode} determines all remaining coefficients order by order.
Because $C^2$ is regular throughout the exterior, one may equivalently construct the series for $Z_1$ or $Q$ first and then multiply by $C^2$.
We use $N_{\rm H}=60$ and $N_\infty=12$.
The solutions are integrated along $\rho=2M+\me^{\mi\theta}\lambda$ from $\lambda_{\rm H}=0.5M$ and $\lambda_\infty=25M$ to $\lambda_{\rm m}=5M$, where $\lambda\in\mathbb R$.
The path angle is fixed at $\theta=-\arg\omega_{\rm ref}+0.50$ during each root search.
This complex contour improves the stability of the integration.

At the matching point, define
$\mathbf{y}_{\rm H}=(Y_s^{\rm H},\partial_\rho Y_s^{\rm H})^{\rm T}$ and
$\mathbf{y}_\infty=(Y_s^\infty,\partial_\rho Y_s^\infty)^{\rm T}$.
The normalized shooting function is
\begin{equation}\label{app:eq:shooting-determinant}
 D_s(\omega)= \frac{\det(\mathbf{y}_{\rm H},\mathbf{y}_\infty)}{\lVert\mathbf{y}_{\rm H}\rVert\lVert\mathbf{y}_\infty\rVert}.
\end{equation}
The condition $D_s=0$ makes the two solutions linearly dependent at the matching point and hence imposes one QNM boundary value problem.
Starting near a reference frequency $\omega_{\rm ref}$, we solve $\Re D_s=0$ and $\Im D_s=0$ simultaneously.

Table~\ref{tab:numerical-qnms} lists the frequencies for $n=0,\cdots,5$.
We define $\Delta\omega=\omega_{\rm num}-\omega_{\rm ref}$, and use L and S for the Leaver calculation and the shooting calculation, respectively.
For the RW and spin-1 branches, the largest Leaver deviations are $1.71\times10^{-5}$ and $2.61\times10^{-9}$.
The corresponding shooting deviations are $1.26\times10^{-12}$ and $3.23\times10^{-11}$.
These calculations directly recover the QNM spectrum of the fourth order Bach equations and verify its invariance under the regular conformal transformation.

\begin{table*}[!htb]
\centering
\caption{Direct solutions of the Bach equations with Leaver's method (L), and shooting results (S) testing the invariance of the QNM spectrum under a conformal transformation.
Frequencies are given as $M\omega$.
The reference values are obtained with Leaver's method from the RW equation and the massless vector perturbation equation for a Schwarzschild BH.}
\label{tab:numerical-qnms}
\begin{ruledtabular}
\begin{tabular}{cccccc}
$n$ & Method & $M\omega_{\rm RW}$ & $|M\Delta\omega_{\rm RW}|$
& $M\omega_{s=1}$ & $|M\Delta\omega_{s=1}|$\\
\hline
$0$ & L & $0.373671955611070-0.088962315897942\mi$ & $2.71\times10^{-7}$ & $0.457595511629985-0.095004425819232\mi$ & $2.74\times10^{-13}$\\
    & S & $0.373671684418085-0.088962315688942\mi$ & $4.35\times10^{-14}$ & $0.457595511629847-0.095004425819536\mi$ & $6.46\times10^{-14}$\\
$1$ & L & $0.346711214132316-0.273915011678965\mi$ & $2.57\times10^{-7}$ & $0.436542385745059-0.290710143117316\mi$ & $6.27\times10^{-12}$\\
    & S & $0.346710996879198-0.273914875291250\mi$ & $3.81\times10^{-14}$ & $0.436542385750550-0.290710143120387\mi$ & $2.54\times10^{-14}$\\
$2$ & L & $0.301057020851047-0.478273286037883\mi$ & $5.14\times10^{-6}$ & $0.401186733757459-0.501587346526890\mi$ & $2.44\times10^{-10}$\\
    & S & $0.301053454612418-0.478276983223056\mi$ & $1.14\times10^{-13}$ & $0.401186733916477-0.501587346341646\mi$ & $4.55\times10^{-14}$\\
$3$ & L & $0.251518955320228-0.705142753397113\mi$ & $1.50\times10^{-5}$ & $0.362595033663811-0.730198515318030\mi$ & $1.69\times10^{-9}$\\
    & S & $0.251504962185611-0.705148202433424\mi$ & $4.35\times10^{-13}$ & $0.362595032331488-0.730198514286318\mi$ & $3.28\times10^{-13}$\\
$4$ & L & $0.207531645161141-0.946845893746450\mi$ & $1.71\times10^{-5}$ & $0.328736672697582-0.971609376856302\mi$ & $2.61\times10^{-9}$\\
    & S & $0.207514579813023-0.946844890866271\mi$ & $8.48\times10^{-13}$ & $0.328736671108326-0.971609378932334\mi$ & $2.87\times10^{-11}$\\
$5$ & L & $0.169307383986975-1.195614451864460\mi$ & $1.02\times10^{-5}$ & $0.301492994421637-1.219715248781730\mi$ & $1.63\times10^{-9}$\\
    & S & $0.169299403092933-1.195608054135797\mi$ & $1.26\times10^{-12}$ & $0.301492995734610-1.219715249742210\mi$ & $3.23\times10^{-11}$\\
\end{tabular}
\end{ruledtabular}
\end{table*}

\bibliographystyle{apsrev4-1}
\bibliography{references}

\clearpage
\onecolumngrid
\begingroup
\setcounter{equation}{0}
\renewcommand{\theequation}{S\arabic{equation}}
\renewcommand{\theHequation}{S\arabic{equation}}

\section*{Supplemental Material}
\noindent
{\bf \emph{Brans-Dicke Theory.--}}
The main text proves conformal invariance of the axial QNM spectrum in pure Weyl gravity through an algebraic relation between the linearized Bach operators.
For comparison, we consider a Brans-Dicke theory with local Weyl symmetry and prove the conformal invariance of its QNM spectrum directly from the full action.
Here a conformal transformation acts on both the metric and the compensating scalar as a local gauge redundancy.
A suitable combination of the two fields gives conformally invariant master variables, so the second-order Einstein perturbation equations match term by term before and after the transformation.

Consider the action
\begin{equation}\label{supp:eq:action}
 S_{\rm CE}=\int\dif^4x\sqrt{-g}\left[
 \phi^2R+6g^{\mu\nu}(\nabla_\mu\phi)(\nabla_\nu\phi)
 \right].
\end{equation}
The field $\phi$ is a Weyl compensator.
If matter is included, its action must transform with the corresponding conformal weights.
We restrict the discussion to vacuum.
Under
\begin{equation}\label{supp:eq:weyl-transform}
 g_{\mu\nu}\rightarrow C^2(x)g_{\mu\nu}, \quad \phi\rightarrow C^{-1}(x)\phi
\end{equation}
Eq.~\eqref{supp:eq:action} changes at most by a boundary term.
Variations with respect to $g_{\mu\nu}$ and $\phi$ give
\begin{align}
 \mathscr{E}_{\mu\nu}={}&\phi^2G_{\mu\nu}
 +(g_{\mu\nu}\Box-\nabla_\mu\nabla_\nu)\phi^2 + 6\left[\nabla_\mu\phi\nabla_\nu\phi
 -\frac{1}{2}g_{\mu\nu}(\nabla\phi)^2\right]=0,
 \label{supp:eq:metric-eom}\\
 \mathscr{E}_\phi={}&\phi R-6\Box\phi=0.
 \label{supp:eq:compensator-eom}
\end{align}
Taking the trace of Eq.~\eqref{supp:eq:metric-eom} verifies
\begin{equation}\label{supp:eq:noether-identity}
 g^{\mu\nu}\mathscr{E}_{\mu\nu}=-\phi\mathscr{E}_\phi.
\end{equation}
Thus, for $\phi\neq0$, the scalar equation follows from the trace of the Einstein equation.

Although the compensator has derivative terms resembling a kinetic term in Eq.~\eqref{supp:eq:action}, it does not add a physical scalar mode.
Its perturbation and the conformal trace perturbation of the metric parametrize the same local Weyl gauge freedom.
This is also clear at the level of the action.
For any nonzero constant $\phi_0$, define
\begin{equation}\label{supp:eq:invariant-metric}
 \gamma_{\mu\nu}=\left(\frac{\phi}{\phi_0}\right)^2g_{\mu\nu}.
\end{equation}
The metric $\gamma_{\mu\nu}$ is invariant under Eq.~\eqref{supp:eq:weyl-transform}.
Using Eq.~\eqref{supp:eq:invariant-metric} as a field redefinition and applying the conformal transformation of the Ricci scalar in four dimensions gives
\begin{equation}\label{supp:eq:action-identity}
 \sqrt{-g}\left[\phi^2R+6(\nabla\phi)^2\right]
 =\phi_0^2\sqrt{-\gamma}\,R[\gamma]
 +6\partial_\mu\!\left(\sqrt{-g}\,\phi\nabla^\mu\phi\right).
\end{equation}
After dropping the surface term, Eq.~\eqref{supp:eq:action} is the Einstein-Hilbert action written in terms of $\gamma_{\mu\nu}$.
For $\phi\neq0$, one may choose the Einstein gauge $\phi=\phi_0$.
In any other gauge, $g_{\mu\nu}$ and $\phi$ are simply different parametrizations of the same invariant metric $\gamma_{\mu\nu}$.

Start from the static, spherically symmetric seed metric
\begin{equation}\label{supp:eq:seed}
 \dif\hat s^2=-f(r)\dif t^2+\frac{\dif r^2}{f(r)}
 +r^2\dif\Omega_2^2,
 \qquad \hat\phi=\phi_0.
\end{equation}
For any positive function $C^2(r)$, set
\begin{equation}\label{supp:eq:background}
 g_{\mu\nu}=C^2(r)\hat g_{\mu\nu}, \quad \phi=C^{-1}(r)\phi_0.
\end{equation}
If $\hat g_{\mu\nu}$ obeys the vacuum Einstein equation, then $(g_{\mu\nu},\phi)$ automatically satisfies Eqs.~\eqref{supp:eq:metric-eom} and \eqref{supp:eq:compensator-eom}, while Eq.~\eqref{supp:eq:invariant-metric} still gives $\gamma_{\mu\nu}=\hat g_{\mu\nu}$.
The Einstein tensor of $g_{\mu\nu}$ alone is generally nonzero.
The compensator derivatives in Eq.~\eqref{supp:eq:metric-eom} cancel these extra terms exactly, leaving the full field configuration as a vacuum solution.
If the equation is written in the form of an ordinary Einstein equation, these terms may be moved into an effective source,
\begin{equation}\label{supp:eq:effective-source}
 G_{\mu\nu}[g]=-\frac{1}{\phi^2}\left\{
 (g_{\mu\nu}\Box-\nabla_\mu\nabla_\nu)\phi^2
 +6\left[\nabla_\mu\phi\nabla_\nu\phi
 -\frac{1}{2}g_{\mu\nu}(\nabla\phi)^2\right]\right\}
 \equiv T_{\mu\nu}^{(\phi)}.
\end{equation}
The tensor $T_{\mu\nu}^{(\phi)}$ is not the stress tensor of an added scalar matter field.
It changes with the Weyl gauge and vanishes in the Einstein gauge.
The geometry invariant under the gauge transformation is still $\gamma_{\mu\nu}$ defined by Eq.~\eqref{supp:eq:invariant-metric}.

As an example, a conformally regular Schwarzschild BH metric may be defined by
\begin{equation}\label{supp:eq:regular-background}
 f(r)=1-\frac{2M}{r},
 \qquad
 C^2(r)=\left(1+\frac{L^2}{r^2}\right)^{2N},
 \qquad N\in\mathbb N^+.
\end{equation}
The conformal factor diverges at $r=0$ and makes the curvature invariants of $g_{\mu\nu}$ finite.
At the same point the compensator vanishes, so the Einstein gauge degenerates at the center.
Our QNM problem is defined only in the exterior region.
For Eq.~\eqref{supp:eq:regular-background}, both $C$ and $C^{-1}$ are smooth and finite for $r\geq2M$, and $C\to1$ at infinity.

We first test conformal invariance of the perturbation operator with a massless scalar field $\varphi$ satisfying
\begin{equation}\label{supp:eq:probe-scalar}
 \left(\Box-\frac{R}{6}\right)\varphi=0,
 \quad
 \varphi\rightarrow C^{-1}\varphi.
\end{equation}
Let $\mathcal{P}_g=\Box_g-R[g]/6$.
For $g_{\mu\nu}=C^2\hat g_{\mu\nu}$,
\begin{equation}\label{supp:eq:scalar-covariance}
 \mathcal{P}_g\left(C^{-1}\hat\varphi\right) = C^{-3}\mathcal{P}_{\hat g}\hat\varphi.
\end{equation}
The scalar field may therefore be separated as
\begin{equation}\label{supp:eq:scalar-ansatz}
 \varphi=\frac{\me^{-\mi\omega t}}{rC(r)}
 \Phi(r)Y_{\ell m}(\theta,\phi),
\end{equation}
where $\dif r_*/\dif r=f^{-1}(r)$.
Below, a prime denotes $\dif/\dif r$.
Substitution into Eq.~\eqref{supp:eq:probe-scalar} gives
\begin{equation}\label{supp:eq:scalar-wave}
 \frac{\dif^2\Phi}{\dif r_*^2} + \left[\omega^2-V_0(r)\right]\Phi=0,
\end{equation}
with
\begin{equation}\label{supp:eq:scalar-potential}
 V_0=f\left[\frac{\ell(\ell+1)}{r^2}-\frac{f^{\prime\prime}}{6}+\frac{1-f+rf^{\prime}}{3r^2}\right].
\end{equation}
Neither the conformal factor nor its derivatives appear in Eq.~\eqref{supp:eq:scalar-wave}.
For a Schwarzschild BH background,
\begin{equation}\label{supp:eq:schwarzschild-scalar-potential}
	V_0=f\left[\frac{\ell(\ell+1)}{r^2}+\frac{2M}{r^3}\right].
\end{equation}
With the same boundary conditions, $\Phi$ therefore has the same QNM frequencies as a conformally coupled scalar on a Schwarzschild BH.

We next consider axial gravitational perturbations.
In the RW gauge,
\begin{equation}\label{supp:eq:axial-ansatz}
 h_{tA}=\me^{-\mi\omega t}h_0(r)S_A^{\ell m},
 \quad
 h_{rA}=\me^{-\mi\omega t}h_1(r)S_A^{\ell m}.
\end{equation}
The odd parity sector has no scalar harmonic component, so $\delta\phi=0$.
Equation~\eqref{supp:eq:invariant-metric} gives
\begin{equation}\label{supp:eq:invariant-perturbation}
 \delta\gamma_{\mu\nu}=C^{-2}h_{\mu\nu}.
\end{equation}
Thus $C^{-2}h_0$ and $C^{-2}h_1$ are precisely the axial perturbations of the seed metric.
A conformally invariant master variable is
\begin{equation}\label{supp:eq:master}
 Q(r)=\frac{f(r)h_1(r)}{rC^2(r)}.
\end{equation}
Substituting Eq.~\eqref{supp:eq:axial-ansatz} into the linearized Eq.~\eqref{supp:eq:metric-eom}, one component of $\delta\mathscr{E}_{\mu\nu}$ gives
\begin{equation}\label{supp:eq:constraint}
 \mi\omega C^2 h_0+f^2C^2 h_1'
 +fh_1\left[C^2f'-f\left(C^2\right)'\right]=0.
\end{equation}
Equation~\eqref{supp:eq:master} then reconstructs both metric components:
\begin{equation}\label{supp:eq:reconstruction}
 h_1=\frac{rC^2}{f}Q,
\qquad
 h_0=\frac{\mi fC^2}{\omega}\frac{\dif(rQ)}{\dif r}.
\end{equation}
Substituting Eq.~\eqref{supp:eq:reconstruction} into the remaining independent axial equation cancels every term involving $C^2$, $\left(C^2\right)'$, and $\left(C^2\right)''$, leaving
\begin{equation}\label{supp:eq:gravitational-wave}
 \frac{\dif^2 Q}{\dif r_*^2}+\left[\omega^2-V_2(r)\right]Q=0,
\end{equation}
where
\begin{equation}\label{supp:eq:general-gravitational-potential}
 V_2=f\left[\frac{\ell(\ell+1)-2}{r^2}+\frac{2f}{r^2} + \frac{f^{\prime}}{r}+f^{\prime\prime}\right].
\end{equation}
For the Schwarzschild BH metric function in Eq.~\eqref{supp:eq:regular-background},
\begin{equation}\label{supp:eq:rw-potential}
 V_2=f\left[
 \frac{\ell(\ell+1)}{r^2}-\frac{6M}{r^3}
 \right].
\end{equation}
Equation~\eqref{supp:eq:gravitational-wave} is therefore the standard RW equation.
The cancellation of the conformal factor follows from transforming the metric and compensator together.
Equation~\eqref{supp:eq:invariant-perturbation} shows the same result directly: if $\hat h_1$ is an axial perturbation of the seed metric, then $h_1=C^2\hat h_1$, and Eq.~\eqref{supp:eq:master} gives $Q=f\hat h_1/r$.

Identical master equations alone do not guarantee identical QNM spectra.
Because the conformal factor $C^2$ multiplies both $g_{tt}$ and $g_{rr}$, radial null geodesics still obey $\dif t=\pm\dif r/f$, and the tortoise coordinate $r_*$ is unchanged.
If $C$ and $C^{-1}$ are smooth and finite at the horizon and infinity, Eqs.~\eqref{supp:eq:scalar-ansatz} and \eqref{supp:eq:reconstruction} change only the perturbation amplitudes and do not mix $\me^{-\mi\omega r_*}$ with $\me^{\mi\omega r_*}$.
For $u=\Phi$ or $Q$, the boundary conditions are therefore
\begin{equation}\label{supp:eq:qnm-boundaries}
 u\sim\me^{-\mi\omega r_*}\ (r\to r_{\rm h}), \quad u\sim\me^{\mi\omega r_*}\ (r\to\infty).
\end{equation}
With the same normalization of the time coordinate, the ingoing horizon condition and outgoing infinity condition map one to one between the two solution spaces.
The conformally coupled scalar and the axial gravitational perturbations of the Brans-Dicke theory therefore have the same QNM spectra as their Schwarzschild BH counterparts.
If only the metric is transformed while the compensator is fixed, or if the test scalar is minimally coupled, the radial potential generally depends on $C^2$ because the chosen field equation is not conformally invariant.
If the conformal factor vanishes or diverges at a boundary, the perturbation equation may remain conformally invariant while the boundary conditions do not.
Conformal symmetry alone then gives no spectral equivalence.

This conclusion also follows directly from Theorem~\ref{thm:qnm-spectral-symmetry}.
Let $T_C$ and $W_C$ denote the conformal map of all perturbation fields and the corresponding weight on the field equations, respectively.
The boundary analysis above shows that $T_C$ is a bijection between the domains of the two boundary value problems.
Here $T_C$ and $W_C$ are independent of frequency; when $W_C$ is invertible on the corresponding field equation space, Eq.~\eqref{eq:general-resolvent-map} immediately ensures that the QNM frequencies and the corresponding pole orders are unchanged.

The proof in this Supplemental Material differs from that in the main text.
Here Eq.~\eqref{supp:eq:invariant-metric} first reduces the Brans-Dicke theory to Einstein gravity for the gauge invariant metric $\gamma_{\mu\nu}$, after which conformally invariant master variables give the scalar and Regge-Wheeler equations on a Schwarzschild background.
Pure Weyl gravity in the main text contains an additional non-Einstein gravitational sector.
Those degrees of freedom cannot be removed by choosing $\gamma_{\mu\nu}$, so an invertible map must instead be constructed between the two complete Bach perturbation operators.
Although the two proofs use different dynamical variables, their physical content is the same: conformal symmetry at the level of the action becomes conformal invariance of the QNM spectrum only when the conformal transformation acts on all fields and preserves the domain of the boundary value problem.

\vspace{5pt}
\noindent
{\bf \emph{Polar Perturbations.--}}
The axial reduction in the main text uses the identity $\delta R=0$.
For polar perturbations, $\delta R$ is generally nonzero, so the trace part of the Schouten tensor must be retained before performing the spherical harmonic decomposition.
Since $\bar{G}_{\mu\nu}=\hat\Lambda\bar{g}_{\mu\nu}$,
\begin{equation}\label{supp:eq:polar-e-trace}
 \mathcal{E}\equiv\bar{g}^{\mu\nu}\mathcal{E}_{\mu\nu}=-\delta R.
\end{equation}
The linearized Bianchi identity still gives $\bar{\nabla}^{\mu}\mathcal{E}_{\mu\nu}=0$.
To retain the trace of the Schouten tensor, define
\begin{equation}\label{supp:eq:polar-x-e}
 X_{\mu\nu}\equiv\delta P_{\mu\nu}+\frac{\hat\Lambda}{6}h_{\mu\nu}
 =\frac{1}{2}\mathcal{E}_{\mu\nu}-\frac{1}{6}\bar{g}_{\mu\nu}\mathcal{E}
 =\frac{1}{2}\mathcal{E}_{\mu\nu}+\frac{1}{6}\bar{g}_{\mu\nu}\delta R.
\end{equation}
Equation~\eqref{supp:eq:polar-e-trace} and the linearized Bianchi identity imply
\begin{equation}\label{supp:eq:polar-x-constraints}
 X\equiv\bar{g}^{\mu\nu}X_{\mu\nu}
 =-\frac{1}{6}\mathcal{E}=\frac{1}{6}\delta R,
 \quad
 \bar{\nabla}^{\mu}X_{\mu\nu}=\bar{\nabla}_{\nu}X.
\end{equation}
For $\delta R=0$, Eq.~\eqref{supp:eq:polar-x-e} reduces to Eq.~\eqref{app:eq:delta-p-e}.
For polar perturbations, $\mathcal{E}_{\mu\nu}/2$ in Eqs.~\eqref{app:eq:nabla-p} and \eqref{app:eq:weyl-coupling-variation} must be replaced by $X_{\mu\nu}$.
Equation~\eqref{app:eq:bach-schouten} therefore becomes
\begin{equation}\label{supp:eq:polar-bach-x}
 \delta\mathcal{B}_{\mu\nu}
 =\bar{\nabla}^{2}X_{\mu\nu}
 -\bar{\nabla}^{\rho}\bar{\nabla}_{\mu}X_{\rho\nu}
 +\bar{C}_{\mu\rho\nu\sigma}X^{\rho\sigma}.
\end{equation}
Following the axial calculation in Appendix~\ref{app:linearized-bach}, we use Eq.~\eqref{supp:eq:polar-x-constraints} to treat the derivative terms and then Eq.~\eqref{supp:eq:polar-x-e} to express the result in terms of $\mathcal{E}_{\mu\nu}$ and its trace.
The linearized Bach tensor for $\delta R\neq0$ is thus
\begin{equation}\label{supp:eq:polar-bach-e}
 \delta\mathcal{B}_{\mu\nu}=\frac{1}{2}\left[
 \left(\bar{\nabla}^{2}+\frac{4\hat\Lambda}{3}\right)\mathcal{E}_{\mu\nu}
 +2\bar{C}_{\mu\rho\nu\sigma}\mathcal{E}^{\rho\sigma}
 \right]+\frac{1}{6}\left[
 \bar{\nabla}_{\mu}\bar{\nabla}_{\nu}
 -\bar{g}_{\mu\nu}\left(\bar{\nabla}^{2}+\hat\Lambda\right)
 \right]\mathcal{E}.
\end{equation}
The second bracket in Eq.~\eqref{supp:eq:polar-bach-e} is the additional term required in the polar sector.

Pure Weyl gravity also admits the local Weyl transformation
\begin{equation}\label{supp:eq:polar-weyl-metric}
 h_{\mu\nu}\rightarrow h_{\mu\nu}+2\sigma\bar{g}_{\mu\nu}.
\end{equation}
With the conventions of Eq.~\eqref{app:eq:einstein-background}, the scalar curvature perturbation transforms as
\begin{equation}\label{supp:eq:polar-weyl-trace}
 \delta R \rightarrow
 \delta R - 6\left(\bar{\nabla}^{2}-\frac{4\hat\Lambda}{3}\right)\sigma.
\end{equation}
The Weyl gauge $\delta R=0$ can therefore be fixed by choosing $\sigma$ to satisfy
\begin{equation}\label{supp:eq:polar-weyl-gauge}
 \left(\bar{\nabla}^{2}-\frac{4\hat\Lambda}{3}\right)\sigma
 =\frac{1}{6}\delta R.
\end{equation}
Then $\mathcal{E}=0$, and Eq.~\eqref{supp:eq:polar-bach-e} reduces to Eq.~\eqref{eq:linearized-bach-image} in the main text.
For the purely axial perturbations considered there, this Weyl gauge choice is simpler.
The term $2\sigma\bar{g}_{\mu\nu}$ has even parity and generates no axial perturbation components.
The Weyl parameter may therefore be set to $\sigma=0$ within the purely axial sector.
Axial perturbations already satisfy $\delta R=0$, so no nontrivial Weyl transformation is required to impose this condition.

Equation~\eqref{supp:eq:polar-weyl-gauge} fixes only the inhomogeneous part of $\sigma$ and leaves a residual gauge freedom that satisfies the homogeneous equation.
If $\sigma_{\rm p}$ is a particular solution, the general solution is
\begin{equation}\label{supp:eq:polar-residual-weyl}
 \sigma=\sigma_{\rm p}+\sigma_{\rm res},
 \quad
 \left(\bar{\nabla}^{2}-\frac{4\hat\Lambda}{3}\right)
 \sigma_{\rm res}=0.
\end{equation}
Equation~\eqref{supp:eq:polar-weyl-trace} shows that the transformation generated by $\sigma_{\rm res}$ preserves $\delta R=0$.
Defining $\epsilon=-2\sigma_{\rm res}$ and using the definition of $\mathcal{E}_{\mu\nu}$ gives
\begin{equation}\label{supp:eq:polar-residual-e}
 \mathcal{E}_{\mu\nu}\rightarrow\mathcal{E}_{\mu\nu}
 +\left(\bar{\nabla}_{\mu}\bar{\nabla}_{\nu}
 -\frac{\hat\Lambda}{3}\bar{g}_{\mu\nu}\right)\epsilon,
 \quad
 \left(\bar{\nabla}^{2}-\frac{4\hat\Lambda}{3}\right)\epsilon=0.
\end{equation}

We now consider polar perturbations with $\ell\geq2$.
For fixed $(\ell,m)$, the polar part of $\mathcal{E}_{\mu\nu}$ contains seven radial amplitudes in the three tensor blocks $\mathcal{E}_{ab}$, $\mathcal{E}_{aA}$, and $\mathcal{E}_{AB}$, where $a,b=t,\rho$.
In the $\delta R=0$ gauge, tracelessness imposes one constraint.
The $\nu=t,\rho$ components and the polar angular component of the Bianchi identity give three independent constraints, while the residual Weyl freedom in Eq.~\eqref{supp:eq:polar-residual-e} removes one further radial amplitude.
The number of polar degrees of freedom in the non-Einstein sector is therefore
\begin{equation}\label{supp:eq:polar-dof-count}
  N_{\mathcal{E}}^{(+)}=7-1-3-1=2.
\end{equation}
This count applies to the sector with $\mathcal{E}_{\mu\nu}\neq0$.
The Einstein degrees of freedom with $\mathcal{E}_{\mu\nu}=0$ are not included in Eq.~\eqref{supp:eq:polar-dof-count}; they remain as homogeneous solutions when $h_{\mu\nu}$ is reconstructed from $\mathcal{E}_{\mu\nu}$.

This counting establishes two independent dynamical variables in the non-Einstein sector but does not identify their physical content.
Using $\mathcal{E}=0$, the linearized vacuum Bach equation in Eq.~\eqref{eq:linearized-bach-image} can be written as
\begin{equation}\label{supp:eq:polar-pm-equation}
 \left(\bar{\nabla}^{2}+\frac{2\hat\Lambda}{3}\right)
 \mathcal{E}_{\mu\nu}
 +2\bar{R}_{\mu\rho\nu\sigma}\mathcal{E}^{\rho\sigma}=0.
\end{equation}
Equation~\eqref{supp:eq:polar-pm-equation} has the same form as the Fierz-Pauli equation for a spin-2 field in the transverse-traceless gauge on an Einstein manifold \cite{Rosen:2020crj}.
The residual Weyl transformation in Eq.~\eqref{supp:eq:polar-residual-e} has the form of the scalar gauge symmetry of this spin-2 system in the transverse-traceless gauge \cite{Deser:2012qg}.
In the limit $\hat{\Lambda}\to0$, the two dynamical variables reduce to a massless spin-2 field and a massless spin-1 field on the Schwarzschild background \cite{Rosen:2020crj}:
\begin{equation}\label{supp:eq:polar-spin-decomposition}
 \mathcal{E}_{\mu\nu}
 =\chi_{\mu\nu}+2\bar{\nabla}_{(\mu}A_{\nu)}.
\end{equation}
The tensor $\chi_{\mu\nu}$ describes the massless spin-2 degree of freedom and satisfies
\begin{equation}\label{supp:eq:polar-spin2-field}
 \bar{\nabla}^{\mu}\chi_{\mu\nu}=0,\quad
 \chi^{\mu}{}_{\mu}=0,\quad
 \bar{\nabla}^{2}\chi_{\mu\nu}
 +2\bar{R}_{\mu\rho\nu\sigma}\chi^{\rho\sigma}=0.
\end{equation}
Taking the trace of Eq.~\eqref{supp:eq:polar-spin-decomposition} and using Eq.~\eqref{supp:eq:polar-spin2-field} gives $\bar{\nabla}^{\mu}A_{\mu}=0$, which is the Lorenz gauge condition.
Substituting Eq.~\eqref{supp:eq:polar-spin-decomposition} into the Bianchi identity $\bar{\nabla}_{\mu}\mathcal{E}^{\mu\nu}=0$ then gives
\begin{equation}\label{supp:eq:polar-vector-wave}
 \bar{\nabla}^{2}A_{\nu}=0.
\end{equation}
Defining $\mathcal{F}_{\mu\nu}=2\bar{\nabla}_{[\mu}A_{\nu]}$, Eq.~\eqref{supp:eq:polar-vector-wave} is equivalent to the source free Maxwell equation
\begin{equation}\label{supp:eq:polar-maxwell-equation}
 \bar{\nabla}_{\mu}\mathcal{F}^{\mu\nu}=0.
\end{equation}
In this limit, Eq.~\eqref{supp:eq:polar-residual-e} reduces to $\mathcal{E}_{\mu\nu}\rightarrow\mathcal{E}_{\mu\nu}+\bar{\nabla}_{\mu}\bar{\nabla}_{\nu}\epsilon$.
On the vector part, this residual gauge freedom acts as
\begin{equation}\label{supp:eq:polar-vector-residual-gauge}
 A_{\mu}\rightarrow A_{\mu}
 +\frac{1}{2}\bar{\nabla}_{\mu}\epsilon,
 \quad
 \bar{\nabla}^{2}\epsilon=0.
\end{equation}
This is the residual gauge freedom of the Maxwell field.
We construct the polar spin-1 master variable $Z_1^{(+)}$ from the $t\rho$ component of the field strength:
\begin{equation}\label{supp:eq:polar-spin1-master}
 \mathcal{F}_{t\rho}
 =-\frac{\sqrt{\ell(\ell+1)}}{\rho^2}Z_1^{(+)}
 Y_{\ell m}\me^{-\mi\omega t}.
\end{equation}
The Maxwell equation and the Bianchi identity give
\begin{equation}\label{supp:eq:polar-spin1-radial}
 \left[F\left(Z_1^{(+)}\right)^{\prime}\right]^{\prime}
 +\left(\frac{\omega^2}{F}-\frac{\ell(\ell+1)}{\rho^2}\right)Z_1^{(+)}=0,
\end{equation}
where $^{\prime}$ denotes a derivative with respect to $\rho$.
This is $\mathcal{R}_1Z_1^{(+)}=0$, with $\mathcal{R}_1$ defined in Eq.~\eqref{eq:radial-operators} of the main text.

We next consider the spin-2 part of Eq.~\eqref{supp:eq:polar-spin-decomposition}.
Equation~\eqref{supp:eq:polar-spin2-field} is the massless spin-2 equation on the Schwarzschild background.
Using the $2+2$ decomposition of this background, let $x^a=(t,\rho)$ and $x^A=(\theta,\phi)$.
Lowercase Latin indices $a,b$ take the values $t,\rho$, while uppercase Latin indices $A,B$ label the angular coordinates.
For fixed $(\ell,m)$, the polar spherical harmonic decomposition of $\chi_{\mu\nu}$ is
\begin{align}\label{supp:eq:polar-spin2-harmonics}
 \chi_{ab} & =H_{ab}^{\chi}(\rho)Y_{\ell m}\me^{-\mi\omega t},\nonumber\\
 \chi_{aA} & =j_a^{\chi}(\rho)D_A Y_{\ell m}\me^{-\mi\omega t},\nonumber\\
 \chi_{AB} & =\rho^2\left[K_\chi(\rho)\Omega_{AB}Y_{\ell m}
 +G_\chi(\rho)Y_{AB}^{\ell m}\right]\me^{-\mi\omega t}.
\end{align}
Here $\Omega_{AB}$ and $D_A$ are the metric and covariant derivative on the unit sphere, and $Y_{AB}^{\ell m}\equiv[D_A D_B+\ell(\ell+1)\Omega_{AB}/2]Y_{\ell m}$.
The amplitudes $H_{ab}^{\chi}$, $j_a^{\chi}$, $K_\chi$, and $G_\chi$ define the Zerilli-Moncrief gauge invariants $\widetilde{\chi}_{ab}$ and $\widetilde{K}_\chi$ \cite{Martel:2005ir}.
Let $\bar{\nabla}_a$ be the covariant derivative on the two-dimensional orbit space and define $\rho_a\equiv\bar{\nabla}_a\rho$.
The corresponding Zerilli-Moncrief master variable is
\begin{equation}\label{supp:eq:polar-zerilli-variable}
 Z_2^{(+)}=\frac{2\rho}{\ell(\ell+1)}\left[
 \widetilde{K}_\chi+
 \frac{2}{(\ell-1)(\ell+2)+6M/\rho}\left(
 \rho^a\rho^b\widetilde{\chi}_{ab}
 -\rho \rho^a\bar{\nabla}_a\widetilde{K}_\chi
 \right)\right].
\end{equation}
Substituting Eq.~\eqref{supp:eq:polar-spin2-harmonics} into Eq.~\eqref{supp:eq:polar-spin2-field}, the transverse and traceless conditions eliminate the dependent radial amplitudes.
Writing the remaining dynamical equation in terms of $Z_2^{(+)}$ gives
\begin{equation}\label{supp:eq:polar-zerilli-equation}
 \mathcal{R}_{\rm Z}Z_2^{(+)}
 \equiv\left[\partial_{r_*}^2+\omega^2-V_{\rm Z}(\rho)\right]Z_2^{(+)}=0,
\end{equation}
where the Zerilli potential is
\begin{equation}\label{supp:eq:polar-zerilli-potential}
 V_{\rm Z}=\frac{F}{\left[(\ell-1)(\ell+2)+6M/\rho\right]^2}\left[
 (\ell-1)^2(\ell+2)^2
 \left(\frac{\ell(\ell+1)}{\rho^2}+\frac{6M}{\rho^3}\right)
 +\frac{36M^2}{\rho^4}\left((\ell-1)(\ell+2)+\frac{2M}{\rho}\right)
 \right].
\end{equation}
Metric perturbations satisfying $\mathcal{E}_{\mu\nu}[h^{\rm E}]=0$ also satisfy $\delta\mathcal{B}_{\mu\nu}[h^{\rm E}]=0$ and therefore belong to the solution space of the linearized Bach equation.
For an asymptotically flat Schwarzschild background, $\mathcal{E}_{\mu\nu}[h^{\rm E}]=\delta G_{\mu\nu}[h^{\rm E}]=0$ is the standard linearized vacuum Einstein equation.
In the RW gauge, the polar part of $h_{\mu\nu}^{\rm E}$ is described by the Zerilli master variable $\Psi_{\rm E}^{(+)}$.
The standard polar reduction gives \cite{Martel:2005ir}
\begin{equation}\label{supp:eq:polar-einstein-sector}
 \mathcal{R}_{\rm Z}\Psi_{\rm E}^{(+)}=0.
\end{equation}
On the Schwarzschild background, the Zerilli and RW equations are related by the standard Chandrasekhar-Darboux transformation, which preserves the QNM boundary conditions away from the algebraically special frequencies \cite{Glampedakis:2017rar}.
Consequently, $Z_2^{(+)}$ in the non-Einstein sector and $\Psi_{\rm E}^{(+)}$ in the Einstein sector are isospectral to their axial spin-2 counterparts.
The usual polar Zerilli solutions are embedded in the Bach solution space, and the complete polar Bach spectrum consists of the Zerilli and spin-1 branches.
Since the polar and axial spin-1 master variables obey the same Maxwell-type radial equation, the complete polar and axial Bach spectra are isospectral.

\endgroup

\end{document}